\documentclass[english,twoside,a4paper,12pt]{amsart}
\usepackage[T1]{fontenc}
\usepackage[mac]{inputenc}
\usepackage{amsthm}
\usepackage{amsfonts,amscd,amssymb,latexsym}
\theoremstyle{plain}
\newtheorem{theorem}{Theorem}
\newtheorem{lemma}[theorem]{Lemma}

\newtheorem{corollary}[theorem]{Corollary}

\theoremstyle{definition}
\newtheorem{definition}[theorem]{Definition}

\begin{document}

\title{Some results related to the conjecture by Belfiore and Sol\'e}

\newtheorem{summary}[theorem]{Summary}
\author{Anne-Maria Ernvall-Hyt\"onen}
\thanks{Department of Mathematics and Statistics, University of Helsinki, Finland, e-mail: anne-maria.ernvall-hytonen@helsinki.fi. The work was supported by the Academy of Finland.}
\maketitle

\begin{abstract}
In the first part of the paper, we consider the relation between kissing number and the secrecy gain. We show that on an $n=24m+8k$-dimensional even unimodular lattice, if the shortest vector length is $\geq 2m$, then as the number of vectors of length $2m$ decreases, the secrecy gain increases. We will also prove a similar result on general unimodular lattices. We will also consider the situations with shorter vectors. Furthermore, assuming the conjecture by Belfiore and Sol\'e, we will calculate the difference between inverses of secrecy gains as the number of vectors varies. We will show by an example that there exist two lattices in the same dimension with the same shortest vector length and the same kissing number, but different secrecy gains. Finally, we consider some cases of a question by Elkies by providing an answer for a special class of lattices assuming the conjecture of Belfiore and Sol\'e. We will also get a conditional improvement on some Gaulter's results concerning the conjecture.
\end{abstract}

\section{Introduction}
Belfiore and Oggier defined in \cite{belfioggiswire} the secrecy gain
\[
\chi_{\Lambda}=\max_{y\in \mathbb{R}, 0<y}\frac{\Theta_{\mathbb{Z}^n}(yi)}{\Theta_{\Lambda}(yi)},
\]
where
\[
\Theta_{\Lambda}(z)=\sum_{x\in \Lambda}e^{\pi i ||x||^2z}
\]
as a new lattice invariant to measure how much confusion the eavesdropper will experience while the lattice $\Lambda$ is used in Gaussian wiretap coding. The function $\Xi_{\Lambda}(y)=\frac{\Theta_{\mathbb{Z}^n}(yi)}{\Theta_{\Lambda}(yi)}$ is called the secrecy function. Belfiore and Sol\'e then conjectured in \cite{belfisole} that the secrecy function attains its maximum at $y=1$, which would then be the value of the secrecy gain. Ernvall-Hyt\"onen \cite{oma:ieee} proved this for all known (and for some possibly existing) extremal lattices), and derived a method to prove or disprove the conjecture for any given unimodular lattice. The secrecy gain was further studied by Oggier, Sol\'e and Belfiore in \cite{belfisoleoggis}, where the authors showed, for instance, that the secrecy gain grows exponentially with the dimension of the lattice, tied the analysis to the study of the flatness factor, and gave several examples. and by Lin and Oggier in \cite{oglinitw}, where they classified the best codes coming from unimodular lattices in dimensios $8\leq n\leq 16$.

Recently, Lin and Oggier considered unimodular lattices in dimensions up to  $23$ \cite{oggierlin}, and furthermore, they considered the dependence of the secrecy gain on the kissing number $K(\Lambda)$ of the lattice. They proved that in dimensions $16\leq n\leq 23$, for non-extremal lattices the secrecy gain is given by
\[
\chi_{\Lambda}=\frac{1}{1-\frac{2n}{2^6}+\frac{2n(n-23)+K(\Lambda)}{2^{12}}}.
\]
In particular, in this case, they proved that the smaller the kissing number, the better the secrecy gain.

The question whether one can use kissing number to find the best secrecy gain 
in general, is not that straightforward to answer: in some cases, the kissing number determines the secrecy gain, but not always.

We prove that if an even unimodular lattice in dimension $n=24m+8k$ ($k\in \{0,1,2\}$) has the shortest vector length $\geq 2m$, then the secrecy gain increases as the number of vectors of length $2m$ decreases (ie. when the kissing number decreases, or as a limit case: when there are no vectors of length $2m$ but the shortest vector length is $2m+2$, i.e. the lattice is extremal). In particular, this shows that the secrecy gain is better on extremal lattices than on lattices with vectors of length $2m$. This is done without assuming the conjecture by Belfiore and Sol\'e. However, assuming the conjecture, we can even calculate how the secrecy gain varies when the number of shortest vectors varies. We will also consider the case when the shortest vector length is $2m-2$, and illustrate why the question is very difficult in the general case. We will also consider odd lattices: if all vectors are of length at least $\left\lfloor \frac{n}{8}\right\rfloor$, then the secrecy gain increases as the number of vectors of length $\left\lfloor \frac{n}{8}\right\rfloor$ decreases. However, since there are very few lattices satisfying the condition, this theorem is not as interesting as the one concerning even lattices. Of course, it would be possible to prove that if all the vectors are of length $\geq\left\lfloor \frac{n}{8}\right\rfloor+1$, then the secrecy gain is better than if the shortest vector length is $\left\lfloor \frac{n}{8}\right\rfloor$. However, since all the lattices with shortest vector length $\left\lfloor \frac{n}{8}\right\rfloor+1$ are known, this would not give any new information past the comparisons in article \cite{oggierlin} (see \cite{conwayodlyzkosloane} for why there are very few lattices like this).

After this, we proceed to consider a case where the kissing numbers are the same but the secrecy gains different. This is done using even and odd lattices in dimension $40$. This should show that in general, one cannot judge lattices by just looking at the kissing numbers. The question remains open whether this could be done for just even lattices (or just odd lattices).

Finally, we connect the conjecture by Belfiore and Sol\'e and the question by Elkies concerning the existence of $n-8k$ lattices. Namely, assuming the conjecture, we will show that the dimension of $n-8k$ lattices without vectors of length $<k$ is bounded. When $k=3$, the result improves the result by Gaulter \cite{gaulter:characteristic}. When $k\geq 4$, no bounds were known previously.
\section{Preliminaries}
For information on theta function, one can study the book \cite{steinshakarchi} by Stein and Shakarchi. For an extensive source on lattices, one may be referred to the book \cite{conwaysloane} by Conway and Sloane. However, to increase the readability of the current article, we will briefly recall the basic facts.

Define first the following theta functions:
\begin{alignat*}{1}
\vartheta_2(\tau) & =e^{\pi i \tau/4}\prod_{n=1}^{\infty}(1-q^{2n})(1+q^{2n})(1+q^{2n-2})\\
\vartheta_3(\tau) &=\prod_{n=1}^{\infty}(1-q^{2n})(1+q^{2n-1})^2\\
\vartheta_4(\tau) &  =\prod_{n=1}^{\infty}(1-q^{2n})(1-q^{2n-1})^2,
\end{alignat*}
where $q=e^{\pi i\tau}$
Notice that $\vartheta_3$ is the theta function of the lattice $\mathbb{Z}^n$.
A lattice $\Lambda$ is called unimodular if its determinant $=\pm 1$, and the norms are integral, ie, $||x||^2\in \mathbb{Z}$ for all vectors $x\in \Lambda$. Further, it is called even, if $||x||^2$ is even for all $x\in\Lambda$. Otherwise it is called odd. A lattice can be even unimodular only if the dimension is divisible by $8$. Odd unimodular lattices exist in all dimensions: $\mathbb{Z}^n$ is an example of such.

Theta functions of even unimodular lattices in dimension $n=24m+8k$ ($k\in \{0,1,2\}$) can be written as polynomials
\[
\Theta=E_4^{3m+k}+\sum_{j=1}^m b_j E_4^{3(m-j)+k}\Delta^j,
\]
where $E_4=\frac{1}{2}\left(\vartheta_2^8+\vartheta_3^8+\vartheta_4^8\right)$ and $\Delta=\frac{1}{256}\vartheta_2^8\vartheta_3^8\vartheta_4^8$. Here $E_4$ is an Eisenstein series, and $\Delta$ the so called discriminant function. For this, and many other results in the theory of theta functions and lattices, see e.g. \cite{conwayodlyzkosloane}.

Generally, theta functions of unimodular lattices in dimension $n=8\mu+\nu$ ($\nu\in \{0,1,\dots,7\}$) can be written as polynomials:
\[
\Theta_{\Lambda}=\sum_{r=0}^{\mu}a_r\vartheta_3^{n-8r}\Delta_8^r,
\]
where $\Delta_8=\frac{1}{16}\vartheta_2^4\vartheta_4^4$. Notice that this gives an alternative representation for theta functions of even lattices.

We call an even lattice \emph{extremal}, if the shortest vectors are of length $2m+2$. Earlier, the definition of an extremal lattice stated that a lattice is extremal if the shortest vectors are of length $\left\lfloor\frac{n}{8}\right\rfloor+1$. However, Conway, Odlyzko and Sloane were able to show that there are very few extremal lattices with this definition, and they all exist in dimensions $\leq 24$ \cite{conwayodlyzkosloane}. Rains and Sloane \cite{rainssloane} showed that unless $n=23$, the shortest vector must be $2\left\lfloor\frac{n}{24}\right\rfloor+2$. Furthermore, Gaulter \cite{gaulter:minima} showed that if $24\mid n$, then any lattice obtaining this bound must be even.

\section{Increasing the secrecy gain by decreasing the number of short vectors}
In this section, we will prove two theorems, first of which corresponds to even lattices, and the second one to all unimodular lattices.
\begin{theorem}\label{melkein} Let $\Lambda$ be an even unimodular lattice in the dimension $n=24m+8k$ with $k\in \{0,1,2\}$ with the shortest vector length $\geq 2m$. Let $k_{2m}\geq 0$ be the number of vectors of the length $2m$. Then the secrecy gain increases as $k_{2m}$ decreases. Assuming the conjecture by Belfiore and Sol\'e, the difference between the inverses of the secrecy gains of two lattices with $k_{2m}$ and $k_{2m}'$ vectors of length $2m$, respectively, is $(k_{2m}-k_{2m}')\frac{3^k}{4^{6m+k}}$.
\end{theorem}
Notice that the first part of the theorem does not require the conjecture by Belfiore and Sol\'e. That is proved unconditionally. By letting $k_{2m}=0$ in the previous theorem, we have the following special case:
\begin{corollary}
The secrecy gain of an $n$-dimensional extremal even unimodular lattice, when $n=24m+8k$ ($k\in \{0,1,2\}$), is better than the secrecy gain of any even $n$-dimensional unimodular lattice with shortest vector length $2m$.
\end{corollary}
Let us now move to the proof of Theorem \ref{melkein}:
\begin{proof}
The theta series of the lattice can be written as a polynomial of the Eisenstein series $E_4$ and the discriminant function $\Delta$:
\[
\Theta=E_4^{3m+k}+\sum_{j=1}^m b_j E_4^{3(m-j)+k}\Delta^j.
\]
Ernvall-Hyt\"onen showed in \cite{oma:ieee} that the  secrecy gain is the maximal value of
\[
\left((1-z)^{3m+k}+\sum_{j=1}^m \frac{b_j}{256^j}(1-z)^{3(m-j)+k}z^{2j}\right)^{-1}
\]
in the range $z\in \left(0,\frac{1}{4}\right]$. Now we need to find how this expression changes when the kissing number changes.
Write
\begin{equation}\label{amaar}
E_4^h\Delta^j=q^{2j}+a_{h,j,1}q^{2j+2}+a_{h,j,2}q^{2j+4}+\cdots,
\end{equation}
where $q=e^{\pi i}$. If the theta function of an even unimodular $24m+8k$-dimensional ($k\in \{0,1,2\}$) lattice is of the form
\[
1+k_{2m}q^{2m}+\cdots,
\]
then to derive the coefficients $b_j$, we have to solve the following system of equations:
\[
\left\{\begin{array}{r} a_{3m+k,0,1}+b_1=0 \\
a_{3m+k,0,2}+b_1a_{3(m-1)+k,1,1}+b_2=0  \\
a_{3m+k,0,3}+b_1a_{3(m-1)+k,1,2}+b_2a_{3(m-2),2,1}+b_3=0  \\
\cdots \quad \quad \quad \cdots \\
a_{3m+k,0,m-1}+b_1a_{3(m-1)+k,1,m-2}+\cdots \quad \\+b_{m-2}a_{3+k,m-2,1}+b_{m-1}=0\\
a_{3m+k,0,m}+b_1a_{3(m-1)+k,1,m-1}+\cdots \quad \\+b_{m-1}a_{k,m-1,1}+b_m=k_{2m} 
\end{array}\right.
\]
From this system of equations, it is clear that $b_1,b_2,\cdots ,b_{m-1}$ depend only on the coefficients $a_{i,j,h}$ for suitable values of $(i,j,h)$, and not on the number $k_{2m}$. We may write
\[
b_m=k_{2m}-\left(a_{3m+k,0,m}+b_1a_{3(m-1)+k,1,m-1}+\cdots +b_{m-1}a_{k,m-1,1}\right).
\]
Let us now compare the secrecy gains of lattices $\Lambda$ and $\Lambda'$. Assume $k_{2m}>k_{2m}'$.  Denote the corresponding coefficients $b_j$ and $b_j'$, respectively. Then $b_j=b_j'$,when $j<m$, and
\begin{multline*}b_m-b_m'=k_{2m}-\left(a_{3m+k,0,m}+\cdots +b_{m-1}a_{k,m-1,1}\right)\\-\left(k_{2m'}-\left(a_{3m+k,0,m}+\cdots +b_{m-1}a_{k,m-1,1}\right)\right)\\=k_{2m}-k_{2m}'>0.
\end{multline*}
 Now the secrecy function of the lattice $\Lambda$ can be written and estimated in the following way:
\begin{multline*}
\Xi_{\Lambda}(y)=\frac{\vartheta_3^n(y)}{\Theta_{\Lambda}(y)}\\=\left(\left(1-\frac{\vartheta_2^4\vartheta_4^4}{\vartheta_3^8}(y)\right)^{3m+k}+\sum_{j=1}^m\frac{b_j}{256^j}\left(1-\frac{\vartheta_2^4\vartheta_4^4}{\vartheta_3^8}(y)\right)^{3(m-j)+k}\cdot\left(\frac{\vartheta_2^4\vartheta_4^4}{\vartheta_3^8}(y)\right)^{2j}\right)^{-1}\\<\left(\left(1-\frac{\vartheta_2^4\vartheta_4^4}{\vartheta_3^8}(y)\right)^{3m+k}+\sum_{j=1}^m\frac{b_j'}{256^j}\left(1-\frac{\vartheta_2^4\vartheta_4^4}{\vartheta_3^8}(y)\right)^{3(m-j)+k}\cdot\left(\frac{\vartheta_2^4\vartheta_4^4}{\vartheta_3^8}(y)\right)^{2j}\right)^{-1}\\=\Xi_{\Lambda'}(y),
\end{multline*}
which proves the first claim.

Let us now move to the proof of the second claim. Assuming that the secrecy gain conjecture holds, the secrecy gain is obtained at $y=1$. Recall $\frac{\vartheta_2^4\vartheta_4^4}{\vartheta_3^8}(1)=\frac{1}{4}$. Now we just need to calculate the difference $\Xi_{\Lambda}^{-1}(1)-\Xi_{\Lambda'}^{-1}(1)$:
\begin{multline*}
\Xi_{\Lambda}^{-1}(1)-\Xi_{\Lambda'}^{-1}(1)\\=\left(\left(1-\frac{\vartheta_2^4\vartheta_4^4}{\vartheta_3^8}(1)\right)^{3m+k}+\sum_{j=1}^m\frac{b_j}{256^j}\left(1-\frac{\vartheta_2^4\vartheta_4^4}{\vartheta_3^8}(1)\right)^{3(m-j)+k}\cdot\left(\frac{\vartheta_2^4\vartheta_4^4}{\vartheta_3^8}(1)\right)^{2j}\right)\\-\left(\left(1-\frac{\vartheta_2^4\vartheta_4^4}{\vartheta_3^8}(1)\right)^{3m+k}+\sum_{j=1}^m\frac{b_j'}{256^j}\left(1-\frac{\vartheta_2^4\vartheta_4^4}{\vartheta_3^8}(1)\right)^{3(m-j)+k}\cdot\left(\frac{\vartheta_2^4\vartheta_4^4}{\vartheta_3^8}(1)\right)^{2j}\right)\\=\frac{b_m-b_m'}{256^m}\left(1-\frac{\vartheta_2^4\vartheta_4^4}{\vartheta_3^8}(1)\right)^{k}\cdot\left(\frac{\vartheta_2^4\vartheta_4^4}{\vartheta_3^8}(1)\right)^{2m}=(b_m-b_m')\frac{3^{2m}}{4^{6m+k}},\end{multline*}
and now the proof is complete.
\end{proof}

The cases with shorter and shorter minimal vectors get more and more complicated: one can not get such a nice dependance. Let us now state and prove the corresponding theorem in the case where the shortest vector length is $2m-2$. After that, let us briefly discuss the difficulties in the general case.

\begin{theorem}\label{2m-2} Assume that even unimodular lattices $\Lambda$ and $\Lambda'$ in dimension $24m+8k$ do not have vectors of length $<2m-2$. Let the number of vectors of length $2m-2$ and $2m$ be $\kappa_{2m-2}$, $\kappa_{2m}$, $\kappa'_{2m-2}$ and $\kappa'_{2m}$ on lattices $\Lambda$ and $\Lambda'$, respectively. Further assume that $\kappa_{2m-2}<\kappa_{2m}$. Assuming the conjecture by Belfiore and Sol\'e, the secrecy gain of lattice $\Lambda$ is greater than the secrecy gain of $\Lambda'$ if
\[
\kappa_{2m}-\kappa_{2m}'<(\kappa_{2m-2}-\kappa_{2m-2}')(240k-24(m-1)-12^3).
\]
\end{theorem}
\begin{proof} Let us study again the system of equations
\[
\left\{\begin{array}{r} a_{3m+k,0,1}+b_1=0 \\
a_{3m+k,0,2}+b_1a_{3(m-1)+k,1,1}+b_2=0  \\
a_{3m+k,0,3}+b_1a_{3(m-1)+k,1,2}+b_2a_{3(m-2),2,1}+b_3=0  \\
\cdots \quad \quad \quad \cdots \\
a_{3m+k,0,m-2}+b_1a_{3(m-1)+k,1,m-3}+\cdots \quad \\+b_{m-3}a_{6+k,m-3,1}+b_{m-2}=0\\
a_{3m+k,0,m-1}+b_1a_{3(m-1)+k,1,m-2}+\cdots \quad \\+b_{m-2}a_{3+k,m-2,1}+b_{m-1}=k_{2m-2}\\
a_{3m+k,0,m}+b_1a_{3(m-1)+k,1,m-1}+\cdots \quad \\+b_{m-1}a_{k,m-1,1}+b_m=k_{2m} 
\end{array}\right.
\]
Depending on the lattice. $k_{2m}=\kappa_{2m}$ or $\kappa'_{2m}$ and $k_{2m-2}=\kappa_{2m-2}$ or $\kappa'_{2m-2}$. Furthermore, denote by $c_i$ and $c_i'$ the coefficients corresponding to the lattices $\Lambda$ and $\Lambda'$, respectively. Now $c_i=c_i'$ for $1\leq m-2$, $$c_{m-1}-c_{m-1}'=\kappa_{2m-2}-\kappa_{2m-2}'$$ and $$c_m-c_m'=\kappa_{2m}-\kappa_{2m}'-c_{m-1}a_{k,m-1,1}+c_{m-1}'a_{k,m-1,1}.$$ Hence,
\begin{multline*}
\Theta_{\Lambda}-\Theta_{\Lambda'}=(\kappa_{m-1}-\kappa_{m-1}')E_4^{3+k}\Delta^{m-1}+\\\left(\kappa_{2m}-\kappa_{2m}'-c_{m-1}a_{k,m-1,1}+c_{m-1}'a_{k,m-1,1}\right)E_4^k\Delta^m\\=(\kappa_{m-1}-\kappa_{m-1}')E_4^{3+k}\Delta^{m-1}+ \left(\kappa_{2m}-\kappa_{2m}'-(\kappa_{2m-2}-\kappa_{2m-2}')a_{k,m-1,1}\right)E_4^k\Delta^m.\end{multline*}
Now, it is clear that the secrecy gain of the lattice $\Lambda$ is greater than the secrecy gain of the lattice $\Lambda'$ if and only if
\begin{multline*}
\min \kappa_{2m-2}\frac{(1-z)^{3+k}z^{2(m-1)}}{256^{m-1}}+(\kappa_2m-\kappa_{2m-2}a_{k,m-1,1})\frac{(1-z)^kz^{2m}}{256^m}\\<\min \kappa_{2m-2}'\frac{(1-z)^{3+k}z^{2(m-1)}}{256^{m-1}}+(\kappa_2m'-\kappa_{2m-2}'a_{k,m-1,1})\frac{(1-z)^kz^{2m}}{256^m}.\end{multline*}
Assuming the conjecture by Belfiore and Sol\'e, this is equivalent to
\begin{multline}\label{ehto}
\kappa_{2m-2}\cdot 3^{3+k}\cdot 4^{-3-k-2(m-1)-4(m-1)}+(\kappa_{2m}-\kappa_{2m-2}a_{k,m-1,1})\cdot 3^{k}\cdot 4^{-k-2m-4m}\\<\kappa_{2m-2}'\cdot 3^{3+k}\cdot 4^{-3-k-2(m-1)-4(m-1)}+(\kappa_{2m}'-\kappa_{2m-2}'a_{k,m-1,1})\cdot 3^k\cdot 4^{-k-2m-4m}.
\end{multline}
From the definition of $a_{k,m-1,1}$ (see (\ref{amaar})), we obtain
\[
a_{k,m-1,1}=240k-24(m-1).
\]
Substituting this, and simplifying the inequality (\ref{ehto}), we obtain
\[
\kappa_{2m}-\kappa_{2m}'<(\kappa_{2m-2}-\kappa_{2m-2}')(240k-24(m-1)-12^3),
\]
which proves the claim.
\end{proof}

Let us now discuss the general case. Assuming that the latices do not have vectors of length $< 2\ell$, or that the length of the vectors is the same, the theta functions differ only by
\[
\sum_{j=\ell}^{m}b_jE_4^{3(m-j)}\Delta^j.
\]
Hence, it suffices to consider the behavior of the polynomial
\[
\sum_{j=\ell}^m\frac{b_j}{256^j}(1-z)^{3(m-j)+k}z^{2j}.
\]
To further simplify the problem, let us assume the conjecture by Belfiore and Sol\'e. This means that we need to consider expressions
\[
\sum_{j=\ell}^m b_j\cdot 3^{3(m-j)}4^{-3(m-j)-k-2j-4j}
\]
for different lattices. Clearly, it is easy to determine the coefficients $b_j$ for any given lattice. However, in the general case, the expressions for coefficients easily become very technical, and it is very difficult to say anything about their size, for instance, and therefore, it is very difficult to prove that lattices of certain type must be better than lattices of any other type. To further illustrate the problem, consider the expression
\[
E_4^h\Delta^j=q^{2j}+a_{h,j,1}q^{2j+2}+a_{h,j,2}q^{2j+4}+\cdots.
\]
First, if $j=0$, then $E_4^h\Delta^j=E_4^h$, i.e. a power of $E_4$. Since
\[
E_4=1+240\sum_{n\geq 1}\sigma_3(n)q^{2n},
\]
where $\sigma_3(n)=\sum_{\substack{d\mid n\\ d>0}}d^3$, we have
\[
An^3 \leq n^3\leq \sigma_3(n)\leq Bn^3
\]
for some constants $A$ and $B$ (see e.g. \cite{serre}, page 152, Proposition 9), ie. $\sigma_3(n)$ is approximately of the size $n^{3}$, and hence, writing
\[
E_4^h=1+\sum_{n\geq 1}f(n)q^{2n},
\]
we have
\[
f(n)=\sum_{\substack{\ell_1+\ell_2+\cdots \ell_t=n\\t\leq h}}240^t\sigma(\ell_1)\sigma(\ell_2)\cdots \sigma(\ell_t).
\]
Since all the sigmas are positive, there is no cancellation, and hence, we may conclude that the values of $f(n)$ get very large. Furthermore, recall that $\Delta$ is holomorphic cusp form of weight $12$, and therefore, $\Delta^jE_4^h$ is a holomorphic cusp form of weight $12j+4h$. Write
\[
\Delta^jE_4^h=\sum_{n\geq j}b(n)q^{2n}.
\]
By Deligne's result \cite{deligne}, we have $|b(n)|\ll n^{(12j+4h-1)/2}d(n)$. Of course, for many values of $n$, the value of $b(n)$ may be much smaller, although in general, this bound is good. Furthermore, by Rankin's result, \cite{rankin}, there are some positive constants $\delta_1$ and $\delta_2$ such that
\[
x(\log x)^{-\delta_1}\ll\sum_{n=1}^x |\tau(n)n^{-(\kappa-1)/2}|\ll x(\log x)^{-\delta_2},
\]
where $\kappa$ is the weight of the form. Hence, in average, the Fourier coefficients $\tau(n)$ of the $\Delta$ function are in average of the size $n^{(11-1)/2}(\log n)^{-\delta}$ for some suitable value/values of $\delta$. So, even though there is some cancellation, it is pretty clear that the coefficients $a_{h,j,n}$ may get very large by the estimates above, and therefore, it is very difficult, or perhaps even impossible to derive any general principles for the dependence of the secrecy gain on the kissing number, or the number of the shortest vectors. Of course, one may prove similar theorems to Theorem \ref{2m-2}, but one would need much more information or some considerably more sophisticated methods than the ones presented in this article to derive any conclusions whether lattices with greatest kissing numbers or longest shortest vectors are better than other lattices, for instance.

Finally, let us briefly consider odd lattices. One may prove the following theorem:
\begin{theorem}\label{melkein_pariton} Let $\Lambda$ be an odd unimodular lattice in the dimension $n$ with the shortest vector length $\left\lfloor \frac{n}{8}\right\rfloor$. Let the number of vectors of length $\left\lfloor \frac{n}{8}\right\rfloor$ be $h$. Then, when $h$ decreases, the secrecy gain increases. Furthermore, assuming the conjecture by Belfiore and Sol\'e, if lattices $\Lambda$ and $\Lambda'$ have $h$ and $h'$ vectors of length $\left\lfloor \frac{n}{8}\right\rfloor$, respectively, then the difference of the inverses of the secrecy gains is $\frac{h-h'}{4^{5\left\lfloor \frac{n}{8}\right\rfloor}}.$
\end{theorem}

However, since Rains and Sloane \cite{rainssloane} have proved that in dimension $n$ the maximal possible minimal norm is $2\left\lfloor\frac{n}{24}\right\rfloor+2$, also on odd unimodular lattices unless $n=23$, when the maximal minimal norm is $3$, this theorem does not give much information: We can easily derive from the inequality
\[
\left\lfloor\frac{n}{8}\right\rfloor\leq 2\left\lfloor\frac{n}{24}\right\rfloor+2
\]
that $n\leq 55$ (actually, one could exclude also some other values of $n$, either by treating the inequality carefully or by Gaulter's result \cite{gaulter} stating that a lattice in dimension $24m$ meeting the bound $2m+2$ must be even).

Furthermore, one could formulate a similar corollary as in the case of even lattices. However, since all the lattices with shortest vector length $\left\lfloor \frac{n}{8}\right\rfloor +1$ are known, this doesn't give any new information compared to \cite{oggierlin}.

\begin{proof}
Since the proof of Theorem \ref{melkein_pariton} is similar to the proof of Theorem \ref{melkein}, we will only sketch the proof to point the differences.

Since the theta function of any unimodular lattice has the polynomial representation
\[
\Theta_{\Lambda}=\sum_{r=0}^{\left\lfloor \frac{n}{8}\right\rfloor}a_r\vartheta_3^{n-8r}\Delta_8^r,
\]
writing
\[
\vartheta_3^r\Delta_8^s=q^s+c_{r,s,1}q^{s+1}+c_{r,s,2}q^{s+2}+\cdots,
\]
we get the system of equations
\[
\left\{\begin{array}{r}a_0=1\\ a_0c_{n,0,1}+a_1=0\\\cdots\quad\cdots\\a_0c_{n,0,\left\lfloor\frac{n}{8}\right\rfloor}+a_1c_{n-8,1,\left\lfloor\frac{n}{8}\right\rfloor-1}+\cdots+a_{\left\lfloor\frac{n}{8}\right\rfloor}=h.\end{array}\right.
\]
We may now proceed just like in the proof of the previous theorem. To prove the second part of the theorem, we work just like in the proof of the previous theorem. We use the polynomial expression for the inverse of the secrecy function:
\[
\Xi_{\Lambda}^{-1}=\sum_{r=0}^{\mu}\frac{a_r}{16^r}\frac{\vartheta_2^{4r}\vartheta_4^{4r}}{\vartheta_3^{8r}}
\]
and notice that in the difference between the inverses, only the last terms in the polynomials remain. Their difference is
\[
\frac{(h-h')}{16^{\left\lfloor\frac{n}{8}\right\rfloor}}\cdot \left(\frac{\vartheta_2^{4}\vartheta_4^{4}}{\vartheta_3^{8}}\right)^{\left\lfloor\frac{n}{8}\right\rfloor}.
\]
If the conjecture by Belfiore and Sol\'e holds, then the maximum is obtained at $y=1$, i.e. $\frac{\vartheta_2^{4}\vartheta_4^{4}}{\vartheta_3^{8}}=\frac{1}{4}$, then the difference becomes
\[
\frac{h-h'}{4^{5\left\lfloor\frac{n}{8}\right\rfloor}},
\]
which completes the proof.
\end{proof}

\section{Same kissing numbers, different secrecy gains}
We will show that the extremal even unimodular lattice in dimension $40$ with shortest vector length $4$ and kissing number $39600$ has a different secrecy gain than the odd unimodular lattice in dimension $40$ with the same shortest vector length and kissing number. This lattice is also extremal, in the sense that it has the longest shortest vectors.

It was showed in \cite{oma:ieee} that the $40$-dimensional even unimodular lattice satisfies the secrecy gain conjecture. Let us use the methods from there to find the actual value of the secrecy gain.  The extremal even unimodular lattices in dimension $40$ have theta series are of the form
\[
1+39600q^4+\cdots.
\]
The theta function of the lattice can also be written as
\[
E_4^5-1200E_4^2\Delta,
\]
and therefore, the secrecy gain is the maximal value of the function $\left((1-z)^5-\frac{75}{16}z^2(1-z)^2\right)^{-1}$ on the interval $\left(0,\frac{1}{4}\right]$, which is obtained at $z=\frac{1}{4}$, and this value is $\frac{4096}{297}$.

Consider now the odd unimodular lattices in dimension $40$ with theta series of the form
\[
1+39600q^4+1048576q^5+\cdots.
\]
More on these odd lattices can be found at \cite{harada}. Comparing coefficients, we see that if the theta function of a unimodular lattice is of the form
\[
1+39600q^4+1048576q^5,
\]
then the theta function can be represented as (see \cite{harada})
\[
\vartheta_3^{40}-80\vartheta_3^{32}\Delta_8+1360\vartheta_3^{24}\Delta_8^2-2560\vartheta_3^{16}\Delta_8^3+20480\vartheta_3^8\Delta_8^4.
\]
Using the method from \cite{oma:ieee}, we'll see that it is sufficient to show that the polynomial
\[
1-5z+\frac{1360}{16^2}z^2-\frac{2560}{16^3}z^3+\frac{20480}{16^4}z^4
\]
obtains its minimal value on the interval $\left(0,\frac{1}{4}\right]$ at $z=\frac{1}{4}$, and then to compute the inverse of this value.

To show that the polynomial obtains its minimum at $z=\frac{1}{4}$, let us first differentiate it. The derivative is
\begin{multline*}
\frac{5}{8}\left(2z^3-3z^2+17z-8\right)\leq \frac{5}{8}\left(2z^3+17z-8\right)\\\leq \frac{5}{8}\left(\frac{2}{64}+\frac{17}{4}-8\right)<0
\end{multline*}
on the interval $z\in \left(0,\frac{1}{4}\right]$. Hence, the polynomial obtains its minimum at $z=\frac{1}{4}$, and this value is $\frac{301}{4096}$. Hence, the secrecy gain is $\frac{4096}{301}$.

\section{Connection to $n-8k$ lattices}

The following results are probably mostly of theoretical interest. It ties together the question of $n-8k$-lattices and the conjecture by Belfiore and Sol\'e. Namely, it is conjectured \cite{elkies} that the dimension of $n-8k$ lattices is bounded. In the following, we will first consider the $n-24$-lattices, and we have to assume that the lattice has no roots (i.e. no vectors of norm $1$ or $2$). Further assuming the conjecture by Belfiore and Sol\'e, we are able to improve the bound by Gaulter \cite{gaulter:characteristic} that says that all the $n-24$-lattices without vectors of norm $1$, have the dimension 8388630 or smaller. On the other hand, this result implies that if there are $n-24$ lattices without roots in dimensions $ n\leq 8388630$, then the conjecture by Belfiore and Sol\'e fails. The next result is the general proof that the dimension of $n-8k$ lattices with no vectors of length less than $k$ is bounded. This is an improvement of a result by Gaulter which he shows in his thesis \cite{gaulter:thesis}, ie. that the dimension of $n-8k$ lattices with no vectors of length less than $k+1$ is bounded. Finally, we give some examples, i.e., treat the values of $k=4,5$.

Let us first recall the definition of a characteristic vector, and then move to the definition of an $n-8k$-lattice:
\begin{definition} Vector $w\in \Lambda$ ($\Lambda$ is a lattice) is called a characteristic vector if
\[
(w,v)\equiv (v,v)\ (\bmod\ 2),
\]
for all $v\in \Lambda$, where $(,)$ is the inner product
\end{definition}
\begin{definition}
Lattice is called en $n-8k$-lattice if all its characteristic vectors are of length at least $n-8k$.
\end{definition}

Elkies proved that all $n$-lattices must be of form $\mathbb{Z^{\ell}}\oplus \Lambda$ \cite{elkies:n}. Also $n-8$ and $n-16$ lattices are well-known. For a brief history and recollection of results, see e.g. \cite{gaulter:represent}.

\begin{theorem} If the conjecture by Belfiore and Sol\'e holds, then there are no $n-24$-lattices in dimensions $23171 \leq n\leq 8388630$ without roots, ie. without vectors of length $1$ or $2$.
\end{theorem}

\begin{proof} We are going to use the method from \cite{gaulter:characteristic}. First recall (for the argument, see e.g. \cite{nebevenkov}, Section 2) that the theta series of an $n-24$ lattice may can be written as
\begin{equation}\label{poly}
\vartheta_3^n+A\vartheta_3^{n-8}\Delta_8+B\vartheta_3^{n-16}\Delta^2+C\vartheta_3^{n-24}\Delta^3.
\end{equation}
We since the lattice has no roots, we have $A=-2n$ and $B=2n^2-46n$. Next we need some bounds for the value of $C$. Consider the secrecy function
\[
\left(1-\frac{2n}{16}z+\frac{2n^2-46n}{16^2}z^2+\frac{C}{16^3}z^3\right)^{-1}.
\]
First of all, as a ratio of theta functions, the secrecy function must be positive. Hence, in particular,
\[1-\frac{2n}{16}\cdot \frac{1}{4}+\frac{2n^2-46n}{16^2}\cdot \frac{1}{4^2}+\frac{C}{16^3}\cdot \frac{1}{4^3}>0,
\]
and hence,
\[
C>-64^3+2n\cdot 64^2-64(2n^2-46)=-128n^2+64\cdot 175n-64^3.
\]
On the other hand, if the conjecture by Belfiore and Sole holds, it is clear that the value of the secrecy function at $z=\frac{1}{4}$ must be at least as large as the value of the secrecy function at $z=0$. Hence,
\[
1-\frac{2n}{16}\cdot \frac{1}{4}+\frac{2n^2-46n}{16^2}\cdot \frac{1}{4^2}+\frac{C}{16^3}\cdot \frac{1}{4^3}\leq 1.
\] 
Hence,
\[
C\leq -128n^2+174n.
\]
We obtain $C=-128n^2+174n-h$ with $h\in [0,-64^3[$. Gaulter \cite{gaulter:characteristic}, Corollary 3.3 has proved the following result:
\begin{lemma}\label{gaul}
The number of shortest characteristic vectors of a positive definite unimodular $\mathbb{Z}$-lattice of dimension $n$ is at most $2^n$.
\end{lemma}
Furthermore, he has proved that
\begin{equation}\label{ehto}
|\chi_{n-8k}|=\lambda_{\ell}2^{n-8k},
\end{equation}
where $n-8k$ is the length of the shortest characteristic vectors, $\chi_{n-8k}$ is the number of vectors of that length, and $\lambda_{\ell}$ is defined in the following way: The theta-function of the lattice is written as a polynomial
\[
\sum_{\ell=0}^{\left\lfloor\frac{n}{8}\right\rfloor}\lambda_{\ell}\vartheta_3^{n-\ell}E_4^{\ell}.
\]
Since
\[
\Delta_8=\frac{1}{16}\left(\vartheta_3^8-E_4\right),
\]
we may replace the polynomial representation \eqref{poly} by the following representation:
\begin{multline*}
\vartheta_3^n+A\vartheta_3^{n-8}\Delta_8+B\vartheta_3^{n-16}\Delta^2+C\vartheta_3^{n-24}\Delta^3\\=\vartheta_3^n+\frac{A}{16}\vartheta_3^{n-8}(\vartheta_3^8-E_4)+\frac{B}{16^2}\left(\vartheta_3^8-E_4\right)^2+\frac{C}{16}\left(\vartheta_3^8-E_4\right)^3.\end{multline*}
From this we obtain $\lambda_3=-\frac{C}{16^3}$. By \eqref{ehto} and Lemma \ref{gaul}, we have
\[
2^{n}\geq |\chi_{n-24}|=-\frac{C}{16^3}\cdot 2^{n-24},
\]
and hence
\[
2^{36}\geq 128n^2-174n+h,
\]
and in particular,
\[
2^{36}\geq 128n^2-174n,
\]
and hence,
\[
n\leq 23171.
\]

\end{proof}
\begin{theorem} If the conjecture by Belfiore and Sol\'e holds, then the dimension $n$ of $n-8k$ lattices with no vectors shorter than $k$ is bounded, ie. for any $k$ there is an $N_k$ such that there are no $n-8k$ lattices in dimensions $n>N_k$.
\end{theorem}
\begin{proof}
The theta function of the lattice can be written as
\[
\vartheta_3^n+\sum_{\ell=1}^k b_{\ell}\vartheta_3^{n-8\ell}\Delta_8^{\ell}.
\]
The coefficients $b_{\ell}$ for $1\leq \ell\leq k-1$ are determined by the fact that there are no vectors of corresponding length. Furthermore, looking at the equations that determine the coefficients
\[
\left\{\begin{array}{r} c_{n,0,1}+b_1=0\\ c_{n,0,2}+b_1c_{n-8,1,1}+b_2=0\\\cdots\quad\cdots\\c_{n,0,k-1}+b_1c_{n-8,1,k-2}+\cdots+b_{k-1}=0.\end{array}\right.
\]
Furthermore, notice that the coefficients $c_{\ell,j,h}$ are polynomials in $n$. Clearly, they can be exponential in $j$ and $h$, but since we are interested in asymptotics in $n$, and hence we may assume $k$ to be a constant, and now also $j$ and $h$ can be treated as constants, or being below some constant bound, the dependence on $n$ is polynomial.

Consider now the inequalities
\[
1-\sum_{\ell=1}^k\frac{b_{\ell}}{16^{\ell}}\cdot\frac{1}{4^{\ell}}>0
\]
(positivity of the theta function) and
\[
1-\sum_{\ell=1}^k\frac{b_{\ell}}{16^{\ell}}\cdot\frac{1}{4^{\ell}}\leq1
\]
(follows from the conjecture by Belfiore and Sol\'e). From these we see that also $b_k$ must be polynomial in $n$.

Use now again Gaulter's results stating that the number of shortest characteristic vectors is at most $2^n$, and \eqref{ehto}. We need value of $\lambda_k$. By substituting $\Delta_8=\frac{1}{16}\left(\vartheta_3^8-E_4\right)$, we have
\[
\lambda_k=\pm \frac{b_k}{16^k},
\]
and hence, also $\lambda_k$ is polynomial in $n$. Finally, consider the equation \eqref{ehto}
\[
\chi_{n-8k}=\lambda_k2^{n-8k}.
\]
If $\lambda_k$ is negative for values of $n>M_k$ for some $M_k$, then this equation can not hold: the number of vectors cannot be negative, and then we may choose $N_k=M_k$. On the other hand, if $\lambda_k$, is positive, and it's polynomial, then it's increasing. However, we have
\[
2^n\geq \lambda_k2^{n-8k},
\] 
ie.
\[
2^{8k+4k}\geq |b_k|,
\]
which is a contradiction: polynomials are not bounded.
\end{proof}
\begin{theorem}
Assume the Belfiore Sol\'e conjecture. Then there are no $n-32$ lattices without vectors of norm $\leq 3$ in dimensions $\geq 14941$.
\end{theorem}
\begin{proof} We proceed precisely as in the proof of the previous (general) theorem. The theta function of the lattice can be written as
\[
\vartheta^n+\sum_{\ell=1}^4 b_{\ell}\vartheta^{n-8\ell}\Delta_8^{\ell}
\]
with
\begin{alignat*}{1}
b_1&=-2n\\
b_2&=2n^2-46n\\
b_3&=-\frac{4}{3}+92n^2-\frac{4832}{3}n.
\end{alignat*}
Hence
\[
b_4=64^3\cdot 2n-64^2(2n^2-46n)+64\left(\frac{4}{3}n^3-92n^2+\frac{4832}{3}n\right)-s
\]
with $s\in [0,64^4[$. Now
\[
\lambda_4=\frac{b_4}{16^4}.
\]
For any lattice, the following inequality must hold
\[
2^n\geq \lambda_42^{n-32},
\]
ie.
\[
2^{48}\geq 64^3\cdot 2n-64^2(2n^2-46n)+64\left(\frac{4}{3}n^3-92n^2+\frac{4832}{3}n\right)-s\]
This holds only when $n\leq 14940$.
\end{proof}
\begin{theorem}
Assume the Belfiore Sol\'e conjecture. Then there are no $n-40$ lattices without vectors of norm $\leq 3$ in dimensions $\geq 12885$.
\end{theorem}
\begin{proof} We proceed precisely as in the proof of the previous (general) theorem. The theta function of the lattice can be written as
\[
\vartheta^n+\sum_{\ell=1}^5 b_{\ell}\vartheta^{n-8\ell}\Delta_8^{\ell}
\]
with
\begin{alignat*}{1}
b_1&=-2n\\
b_2&=2n^2-46n\\
b_3&=-\frac{4}{3}+92n^2-\frac{4832}{3}n\\
b_4&=\frac{2}{3}n^4-84n^3+\frac{12430}{3}n^2-66542n.
\end{alignat*}
Hence
\begin{multline*}
b_5=64^4\cdot 2n-64^3(2n^2-46n)+64^2\left(\frac{4}{3}n^3-92n^2+\frac{4832}{3}n\right)\\-64\left(\frac{2}{3}n^4-84n^3+\frac{12430}{3}n^2-66542n\right)-s
\end{multline*}
with $s\in [0,64^5[$. Now $\lambda_5=-\frac{b_5}{16^5}$ Again, the following inequality must hold
\[
2^n\geq \lambda_52^{n-40},
\]
i.e.
\begin{multline*}
2^{60}\geq -b_5=-64^4\cdot 2n+64^3(2n^2-46n)-64^2\left(\frac{4}{3}n^3-92n^2+\frac{4832}{3}n\right)\\+64\left(\frac{2}{3}n^4-84n^3+\frac{12430}{3}n^2-66542n\right)+s\\ \geq -64^4\cdot 2n+64^3(2n^2-46n)-64^2\left(\frac{4}{3}n^3-92n^2+\frac{4832}{3}n\right)\\+64\left(\frac{2}{3}n^4-84n^3+\frac{12430}{3}n^2-66542n\right).\end{multline*}
We get $n\leq 12884$.
\end{proof}


\section{Conclusions}
Clearly, there is a dependence between the kissing number and the secrecy gain in some cases, but in the general case, there is not. Of course, it would be interesting to see, whether there would be a dependence in the general case, if the consideration were limited to only even or odd unimodular lattices. Secondly, the conjecture by Belfiore and Sol\'e can be used to solve a question by Elkies in some special cases, namely, when the shortest vectors of an $n-8k$ lattice are of the length $k$.


\section*{Acknowledgment}

The author would like to thank Prof. Frederique Oggier and Fuchun Lin for an inspiring conversation.


\begin{thebibliography}{1}

\bibitem{belfioggiswire}
J.-C. Belfiore and F.~E. Oggier.
\newblock Secrecy gain: A wiretap lattice code design.
\newblock In {\em ISITA}, pages 174--178, 2010.

\bibitem{belfisole}
J.-C. Belfiore and P.~Sol{\'e}.
\newblock Unimodular lattices for the gaussian wiretap channel.
\newblock {\em CoRR}, abs/1007.0449, 2010.

\bibitem{deligne}
P.~Deligne.
\newblock {La conjecture de {W}eil. {I}},
\newblock {\em Institut des Hautes \'Etudes Scientifiques. Publications Math\'ematiques}, 43:273--307, 1974.

\bibitem{elkies}
N.~D. Elkies
\newblock Lattices and codes with long shadows. 
\newblock{\em Mathematical Research Letters} 2 (1995), no. 5, 643–651. 

\bibitem{harada}
S.~Bouyuklieva, I.~Bouyukliev, and M.~Harada.
\newblock Some extremal self-dual codes and unimodular lattices in dimension
  40.
\newblock arXiv:1111.2637].

\bibitem{conwayodlyzkosloane}
J.~H. Conway, A.~M. Odlyzko, and N.~J.~A. Sloane.
\newblock Extremal self-dual lattices exist only in dimensions {$1$} to {$8$},
  {$12$}, {$14$}, {$15$}, {$23$}, and {$24$}.
\newblock {\em Mathematika}, 25(1):36--43, 1978.

\bibitem{conwaysloane}
J.~H. Conway and N.~J.~A. Sloane.
\newblock {\em Sphere packings, lattices and groups}, volume 290 of {\em
  Grundlehren der Mathematischen Wissenschaften [Fundamental Principles of
  Mathematical Sciences]}.
\newblock Springer-Verlag, New York, 1988.
\newblock With contributions by E. Bannai, J. Leech, S. P. Norton, A. M.
  Odlyzko, R. A. Parker, L. Queen and B. B. Venkov.
  
\bibitem{gaulter:minima}
M. Gaulter.
\newblock  Minima of odd unimodular lattices in dimension {$24m$}.
\newblock {\em Journal of Number Theory}, 91(1):81--91, 2001.

\bibitem{gaulter:characteristic}
M. Gaulter
\newblock Lattices without short characteristic vectors
\newblock {\em Mathematical Research Letters} 5, 353--362 (1998).

\bibitem{gaulter:represent}
M. Gaulter.
\newblock Characteristic vectors of unimodular lattices which represent two. (English, French summary) 
\newblock{\em Journal de Théorie des Nombres de Bordeaux} 19 (2007), no. 2, 405–414. 

\bibitem{gaulter:thesis}
M. Gaulter.
\newblock{\em Characteristic Vectors of Unimodular Lattices Over the Integers.}
\newblock Thesis.
\newblock University of California, Santa Barbara, 1998.

\bibitem{oma:ieee}
A.-M. Ernvall-Hyt\"onen.
\newblock On a conjecture by {B}elfiore and {S}ol\'e on some lattices.
\newblock {\em IEEE Transactions on Information Theory}, 58:5950--5955, 2012.

\bibitem{oggierlin}
F.~Lin and F.~Oggier.
\newblock A classification of unimodular lattice wiretap codes in small
  dimensions.
\newblock arXiv:1201.3688.

\bibitem{oglinitw}
F.~Lin and F.~Oggier.
\newblock Secrecy gain of gaussian wiretap codes from unimodular lattices.
\newblock In {\em IEEE ITW}, pages 718--722, 2011.

\bibitem{nebevenkov}
G. Nebe and B. Venkov
\newblock Unimodular lattices with long shadow. 
\newblock {\em Journal of Number Theory} 99 (2003) 307-317.

\bibitem{belfisoleoggis}
F.~Oggier, P.~Sol{\'e}, and J.-C. Belfiore.
\newblock Lattice codes for the wiretap gaussian channel: Construction and
  analysis.
\newblock arXiv:1103.4086

\bibitem{rankin}
Rankin, R. A.(4-GLAS)
\newblock Sums of powers of cusp form coefficients. II. 
\newblock{\em Mathematische Annalen}, 272 (1985), no. 4, 593–600.

\bibitem{rainssloane}
E.~M. Rains, and N.~J.~A. Sloane.
\newblock {The shadow theory of modular and unimodular lattices},
\newblock {\em Journal of Number Theory}, 73(2):359--389, 1998.

\bibitem{serre}
J.-P. Serre.
\newblock{\em Cours d'arithm\'etique.}
\newblock Presses universitaires de France, 2011.

\bibitem{steinshakarchi}
E.~M. Stein and R.~Shakarchi.
\newblock {\em Complex analysis}.
\newblock Princeton Lectures in Analysis, II. Princeton University Press,
  Princeton, NJ, 2003.\end{thebibliography}
\end{document}